\documentclass[11pt]{article}

\usepackage[preprint]{acl}

\usepackage{times}
\usepackage{latexsym}

\usepackage[T1]{fontenc}
\usepackage[utf8]{inputenc}

\usepackage{microtype}

\usepackage{inconsolata}

\usepackage{graphicx}

\newcommand{\myProj}{P$^2$RAG}

\usepackage{amsmath}
\usepackage{amssymb}
\usepackage{bm}
\def\eqref#1{equation~\ref{#1}}
\def\1{\bm{1}}

\def\va{{\bm{a}}}
\def\vb{{\bm{b}}}

\def\vd{{\bm{d}}}

\def\vp{{\bm{p}}}

\def\vx{{\bm{x}}}

\DeclareMathAlphabet{\mathsfit}{\encodingdefault}{\sfdefault}{m}{sl}
\SetMathAlphabet{\mathsfit}{bold}{\encodingdefault}{\sfdefault}{bx}{n}

\def\sF{{\mathbb{F}}}
\def\sG{{\mathbb{G}}}

\renewcommand{\cite}[1]{\citep{#1}}
\newcommand{\yrcite}[1]{\citet{#1}}

\usepackage{booktabs}

\usepackage[capitalize,noabbrev]{cleveref}

\usepackage{amsthm}
\theoremstyle{plain}
\newtheorem{theorem}{Theorem}[section]

\newtheorem{claim}{Claim}[section] % MODIFIED
\theoremstyle{definition}
\newtheorem{definition}[theorem]{Definition}

\theoremstyle{remark}

\usepackage{algorithm}
\usepackage{algorithmic}

\usepackage{etoolbox}
\AtBeginEnvironment{algorithmic}{\small}

\usepackage{enumitem}
\setlist[itemize]{leftmargin=*}
\setlist[enumerate]{leftmargin=*}
\setlist{noitemsep}
\setlist{nolistsep}

\newcommand{\revision}[1]{}
\newcommand{\original}[1]{}
\newcommand{\reviewer}[1]{}

\makeatletter
\edef\mathbreak@comma{\mathchar\number\mathcode`\,\relax}
\edef\mathbreak@lparen{\mathchar\number\mathcode`(\relax}
\edef\mathbreak@rparen{\mathchar\number\mathcode`)\relax}
\edef\mathbreak@lparen@code{\the\mathcode`(}
\@ifpackageloaded{amsmath}{%
  \def\resetMathstrut@{%
    \begingroup
    \setbox\z@\hbox{%
      \mathchardef\@tempa\mathbreak@lparen@code\relax
      \def\@tempb##1"##2##3{\the\textfont"##3\char"}%
      \expandafter\@tempb\meaning\@tempa \relax
    }%
    \edef\@tempa{%
      \ht\Mathstrutbox@\the\ht\z@\relax
      \dp\Mathstrutbox@\the\dp\z@\relax}%
    \expandafter\endgroup\@tempa
  }%
}{}
\begingroup\lccode`\~=`\,\lowercase{\endgroup\def~}{\mathbreak@comma\allowbreak}
\begingroup\lccode`\~=`(\lowercase{\endgroup\def~}{\mathbreak@lparen\allowbreak}
\begingroup\lccode`\~=`)\lowercase{\endgroup\def~}{\mathbreak@rparen\allowbreak}
\mathcode`\,="8000
\mathcode`(="8000
\mathcode`)="8000
\makeatother

\title{{\myProj}: Efficient Privacy-Preserving RAG Service \\
Supporting Arbitrary Top-$k$ Retrieval}

\author{
 \textbf{Yulong Ming\textsuperscript{1}},
 \textbf{Mingyue Wang\textsuperscript{2}},
 \textbf{Jijia Yang\textsuperscript{1}},
 \textbf{Jie Xu\textsuperscript{1}},
\\
 \textbf{Zihan Wu\textsuperscript{1}},
 \textbf{Cong Wang\textsuperscript{1}},
 \textbf{Xiaohua Jia\textsuperscript{1}}
\\
\\
 \textsuperscript{1}City University of Hong Kong,
 \textsuperscript{2}Peng Cheng Laboratory
\\
 \small{
   \textbf{Correspondence:} \href{mailto:myl.7@my.cityu.edu.hk}{myl.7@my.cityu.edu.hk}
 }
}

\begin{document}
\maketitle

% ADDED
\begin{abstract}
Retrieval-Augmented Generation (RAG) enables large language models to use external knowledge, but outsourcing the RAG service raises privacy concerns for both data owners and users.
Privacy-preserving RAG systems address these concerns by performing secure top-$k$ retrieval, which is typically implemented using secure sorting to identify relevant documents.
However, existing systems face challenges supporting arbitrary $k$ due to their inability to change $k$, new security issues, and in particular, efficiency degradation with large $k$.
This is a significant limitation because applications such as finance, law, and healthcare require a $k$ that is large enough to cause huge overhead for existing systems. Also, modern long-context models generally achieve higher accuracy with larger retrieval sets.
We propose {\myProj}, an efficient privacy-preserving RAG service that supports arbitrary top-$k$ retrieval.
Unlike existing systems, {\myProj} avoids sorting candidate documents.
Instead, it uses an interactive bisection method to determine the set of top-$k$ documents.
For security, {\myProj} uses secret sharing on two semi-honest non-colluding servers to protect the data owner's database and the user's prompt.
It enforces restrictions and verification to defend against malicious users and tightly bounds the information leakage of the database.
The experiments show that {\myProj} is 3--300$\times$ faster than the state-of-the-art PRAG for $k = 16$--$1024$.
\end{abstract}

\section{Introduction}

Large Language Models (LLMs) are powerful tools, but they suffer from limitations, such as hallucinations \cite{hallucination} and lack of real-time or domain-specific data.
Retrieval-Augmented Generation (RAG) \cite{rag} provides a powerful solution to these issues without the high cost and complexity of fine-tuning.
By retrieving relevant information from an external knowledge base and combining it with the user's prompt, RAG enables the model to generate responses that are accurate, up-to-date, and factually grounded \cite{rag_example,jie}.
\looseness=-1

RAG as a Service (RaaS) has recently emerged as a popular paradigm.
In this architecture, the data owner outsources its proprietary database to the RAG service.
The user sends a prompt to this service to retrieve the top-$k$ relevant documents from the database.
The user then submits the prompt, augmented by these documents, to the model to generate responses.
However, this workflow raises privacy issues.
The privacy of both the data owner and the user depends on the RAG service's honesty, as it can access the data owner's database and the user's prompt.
Moreover, malicious users can attempt to extract proprietary information from the data owner's database.
\looseness=-1

Privacy-preserving RAG has been proposed to address these issues.
Existing systems treat privacy-preserving RAG as a secure top-$k$ retrieval problem.
They use techniques from secure $k$-Approximate Nearest Neighbor ($k$-ANN) or secure sorting to compute similarity scores, sort the candidates, and select the top-$k$ documents.
\looseness=-1

However, existing systems face challenges supporting arbitrary $k$.
While some do not support dynamic adjustments to $k$ for user queries or raise security issues, most suffer from efficiency degradation with a large $k$ \cite{pann,prag}.
In deployment scenarios such as finance, law, and healthcare, engineering best practice suggests $k = 100$--$150$ \cite{ragk}, causing huge overhead for existing privacy-preserving RAG systems.
Meanwhile, recent studies show that within a particular yet large context threshold, such as 64K tokens, RAG with a large $k$ generally outperforms RAG with a smaller $k$ for modern long-context models \cite{lc_good,lc_good1,lc_good2}.
This finding implies that $k$ can be increased in RAG, which trades model inference efficiency for accuracy.
A large $k$ is also required by applications such as cache-augmented retrieval \cite{cag,cag1}.

\textbf{System overview.}
We propose {\myProj} to address this challenge.
Unlike existing systems, {\myProj} avoids sorting candidate documents.
Instead, we use an interactive bisection method between the user and the servers to determine a threshold.
We identify documents within the threshold as the top-$k$ results.
This design enables the user to select an arbitrary $k$ by selecting the threshold.
Moreover, for RAG applications, the user only requires the set of the top-$k$ documents rather than their specific internal order.
Therefore, with this design, the documents within or outside of the determined threshold remain unsorted.
This method decreases the number of comparisons.
Because comparisons are the primary computation and communication bottleneck in secure protocols, {\myProj} reduces these costs.

{\myProj} protects both the data owner's database and the user's prompt.
We use secret sharing to ensure that the semi-honest RAG service cannot extract the database or the prompt during the workflow.
{\myProj}'s protocol runs on two semi-honest non-colluding servers where the database and the prompt are secret-shared.
Neither server can know any information about the database or the prompt.
Moreover, to defend against malicious users, we design restrictions and verifications for the interactions between servers and users.
We limit the number of result documents, limit the number of bisection iterations, and verify that the final retrieval of textual documents matches the bisection results.
The amount of leakage that a malicious user can extract from the database is tightly bounded.

In summary, our contributions are as follows:

\begin{itemize}
    \item We propose {\myProj}, a privacy-preserving RAG system that supports arbitrary top-$k$ retrieval.
    {\myProj} protects both the data owner's database and the user's prompt.
    {\myProj} uses an interactive bisection method to determine the top-$k$ set without sorting all candidate documents, enabling efficient retrieval even for a large $k$.
    \item We design {\myProj}'s security protocol to defend against both the semi-honest RAG service and malicious users.
    We use secret sharing to protect data privacy against semi-honest servers.
    We enforce restrictions and verification mechanisms to defend against malicious users and tightly bound the information leakage of the database.
    \item We implement {\myProj} and evaluate its performance.
    The code is at
    % \url{https://anonymous.4open.science/r/p2rag}
    \url{https://github.com/myl7/p2rag}
    and licensed under Apache License, Version 2.0.
    The experiments show that {\myProj} is efficient, achieves higher performance with larger $k$, and is 3--300$\times$ faster than the state-of-the-art system, PRAG, for $k = 16$--$1024$.
    \looseness=-1
\end{itemize}

\section{Problem Formulation}
\label{sec:problem}

\subsection{System Model}
\label{sec:problem_system}

{\myProj} aims to provide privacy-preserving RAG as a service to the data owner and users.
During the offline (i.e., preprocessing) stage, the data owner outsources its data to {\myProj} as an encrypted database.
During the online stage, users use encrypted prompts to query top-$k$ relevant documents from the encrypted database.

The data owner's database contains $N$ documents.
Each textual document corresponds to an $m$-dimensional document embedding, i.e., a vector.
Each embedding is generated from the text using an embedding model.
The embedding semantically describes the corresponding text, so that the distance between two embeddings can quantify the similarity of the two texts.
We use the cosine distance as the distance metric, which is widely used in existing RAG applications.
We assume all embeddings are $\ell_2$-normalized.

{\myProj} consists of two servers.
We choose the two-server model over more servers because we can use the efficient cryptographic tools optimized for two servers.
To outsource the database to the servers, the data owner secret-shares it.
That is, for each document embedding, the two servers hold the data owner's two shares, respectively.
The sum of the two shares is the embedding, and a single share reveals no information about it.
Both texts and embeddings are secret-shared in this way.

A user has a textual prompt and a prompt embedding.
The prompt embedding is generated from the text using the same embedding model.
We assume that the embedding model is public.
We also assume the user can get the prompt embedding using the model in advance, as in prior privacy-preserving RAG work \cite{prag,remoterag}.
To use {\myProj}, the user secret-shares its prompt embedding and sends the two shares to the two servers, respectively.
By running {\myProj}'s protocol, the user finally receives the indices of $k$ documents from the servers.
These documents have the $k$ highest cosine similarities to the prompt text.
\looseness=-1

\subsection{Threat Model}
\label{sec:problem_threat}

The two servers are semi-honest and non-colluding.
That is, the servers honestly follow {\myProj}'s protocol but passively attempt to infer information about the database and prompts, i.e., honest-but-curious.
The servers do not share any extra information beyond what is defined in {\myProj}'s protocol.

This assumption is common and widely used in other secret-sharing-based security systems \cite{non_colluding,ppi1}.
We quote a statement by Microsoft \cite{non_colluding_setting} to show some example deployment scenarios: 1) collusion may be physically infeasible, too costly, prevented by law, or blocked by conflicting business interests between the parties; and 2) when independent attackers separately compromise different systems, they may lack the capacity or opportunity to coordinate with each other.
For realistic product deployments, Signal's SecureValueRecovery2 uses hardware enclaves from different platforms for this model, avoiding trusting only one hardware manufacturer. Coinbase's WaaS uses its server and the user's device for this model, keeping secure when the user's device is compromised \cite{non_colluding_deployment}.

Users are malicious. Users can deviate from {\myProj}'s protocol arbitrarily to extract information from the database.
Users do not trust the servers. Data owners are honest. They do not trust any server or any user.
\looseness=-1

\subsection{Security Goals}
\label{sec:problem_goal}

We aim to protect the privacy of both the data owner and the user, i.e., protect both the database and the prompt, resulting in the following goals.

\begin{itemize}
    \item \textit{Privacy}. No server can learn any information about any prompt or document. In particular, because the returned $k$ documents are the $k$ most similar to the prompt and thus can reveal information about the prompt, we must also protect them from the servers.
    \item \textit{Bounded database leakage}. No user can learn any extra information about any document beyond the ``baseline'' leakage and a small amount of leakage defined in {\myProj}'s protocol.
    The ``baseline'' leakage refers to the $k$ documents returned to the user, which is necessary for {\myProj}'s functionality.
    {\myProj} also leaks counts of documents that are in a particular range to accelerate {\myProj}'s protocol.
    This aggregate leakage does not identify any particular document, and the total number of leaked counts is limited, resulting in $O(\log^2 N)$ leakage.
\end{itemize}

\section{Preliminary}
\label{sec:preliminary}

We use \yrcite{secret_share} secret sharing over a prime field $\sF_p$ to protect a value.
We denote both shares of a value $x$ by $[x]$.
The shares are held by the servers, respectively.
We denote each share by $[x]_i$ for $i \in \{0, 1\}$.
That is, we protect a value $x$ by secret-sharing it as $[x]_0 + [x]_1 = x$ and sending $[x]_i$ to the server $i$, respectively.
A single share $[x]_i$ does not reveal any information about $x$.
We have $[x \pm y] = [x] \pm [y]$.
That is, for addition and subtraction, each server only needs to compute on its local shares, and no interaction is required.

Distributed Comparison Functions (DCFs) \cite{fss0,dpf,dcf_cmp,dcf} are schemes to secret-share a comparison function.
Each function share can be individually executed on a server.
The outputs of the function shares are the shares of the output of the original function.

\begin{definition}[Comparison Functions]
    For the input domain $\sG_{in} = \{0, \cdots, p - 1\}$, a group $(\sG_{out}, +)$ as the output domain, $a \in \sG_{in}$, and $b \in \sG_{out}$, a comparison function $f^<_{a, b}$ is a function that for any input $x$, the output $y$ has $y = b$ only when $x < a$, otherwise $y = 0$.
\end{definition}

\begin{definition}[Distributed Comparison Functions]
    \label{def:dcf}
    For the input domain $\sG_{in} = \{0, \cdots, p - 1\}$, the output domain $(\sG_{out}, +)$, $a \in \sG_{in}$, $b \in \sG_{out}$, and a security parameter $\lambda$, a DCF is a scheme consisting of the methods:

    \begin{itemize}
        \item Key generation: $Gen(1^\lambda, a, b) \rightarrow (k_0, k_1)$.
        \item Evaluation: $Eval(k_i, j) \rightarrow y_{i,j}$ for any $i \in \{0, 1\}$ and any $j \in \sG_{in}$.
    \end{itemize}

    That satisfies:

    \begin{itemize}
        \item Correctness: $y_{0,j} + y_{1,j} = b$ only when $j < a$ , otherwise $y_{0,j} + y_{1,j} = 0$.
        \item Privacy: Neither $k_0$ nor $k_1$ reveals any information about $a$ and $b$.
        Formally speaking, there exists a Probabilistic Polynomial Time (PPT) simulator $Sim_{dcf}$ that can generate output computationally indistinguishable from any strict subset of the keys output by $Gen$.
    \end{itemize}
\end{definition}

We use DCFs and \yrcite{dcf_cmp}'s interval containment gate to check if a secret-shared value $[x]$ is in an interval $[x_l, x_r)$ as \cref{alg:cmp}.
Its output is $[1]$ when $x \in [x_l, x_r)$, otherwise $[0]$.
Note that this algorithm handles the addition in $\sG_{in}$, which is an integer addition modulo $p$.
\looseness=-1

\begin{algorithm}[htb]
    \caption{Comparison: \small $[0/1] \gets Cmp([x], [x_l, x_r])$
    }
    \label{alg:cmp}
    \begin{algorithmic}
        \STATE \textbf{The key generation method} $Cmp.Gen([x_l, x_r))$:
        \STATE Sample random $r \in \sF_p$ and distribute $[r]$.
        \STATE $x'_l = x_l + r$, $x'_r = x_r + r$.
        \STATE $(k_0^l, k_1^l) \gets Gen^<(1^\lambda, x'_l, p - 1)$. Distribute $(k_0^l, k_1^l)$.
        \STATE $(k_0^r, k_1^r) \gets Gen^<(1^\lambda, x'_r, 1)$. Distribute $(k_0^r, k_1^r)$.
        \STATE Sample random $[w] \in \sF_p$ s.t. $[w]_0 + [w]_1 = 1\{x'_l > x'_r\}$. Distribute $[w]$.
        \STATE
        \STATE \textbf{The evaluation method} $Cmp.Eval([x])$:
        \STATE Publish $x + r$. $x$ is masked by $r$ and thus protected.
        \STATE $[y]_i^l = Eval(i, k_i^l, x)$, $[y]_i^r = Eval(i, k_i^r, x)$.
        \STATE $[y]_i = [y]_i^l + [y]_i^r + [w]_i$. Output $[y]$.
    \end{algorithmic}
\end{algorithm}

We assume that there is a trusted dealer that generates the shared random values required by {\myProj}, which is common in security systems \cite{prod,trusted_dealer}.
The trusted dealer runs during the offline stage and does not participate in the online stage, so we do not benchmark its performance.
For example, the data owner can be the trusted dealer.
There are also other techniques that can achieve it \cite{trusted_dealer_alt,trusted_dealer_alt1,jie1}.
\looseness=-1

\section{System Design}
\label{sec:design}

A user willing to use {\myProj} first secret-shares its prompt embedding as $[\vp]$.
It then sends the two shares to the two servers, respectively.
The servers first run the distance calculation to calculate the distances $[d_j]$ between the prompt embedding $[\vp]$ and each document embedding $[\vx_j]$.
The servers then run the interactive distance bisection method with the user.
In each iteration, the servers return the count $[c]$ of documents with distances less than a user-specified distance threshold $[d_k]$.
The user then updates $d_k$ by comparing the returned number $c$ with the target $k$.
When $c$ is close enough, as defined by {\myProj}, or the iteration number reaches a server-specified limit, the servers return a secret-shared $0/1$ array $\vd_c$ where the elements corresponding to the in-range documents are set to $1$.
The number of these elements, i.e., $c$, is also limited by a server-specified threshold.
Finally, during the text retrieval, the user retrieves the indices of all $1$ elements without letting the servers know about the retrieved indices.
The servers also verify that the user only retrieves the indices that match the $0/1$ array $\vd_c$.
The workflow is shown as \cref{fig:workflow}.
The full protocol is summarized in \cref{app:protocol}, with theoretical performance analysis in \cref{app:analysis_perf}.

\begin{figure*}[t]
    \begin{center}
    \includegraphics[width=\textwidth]{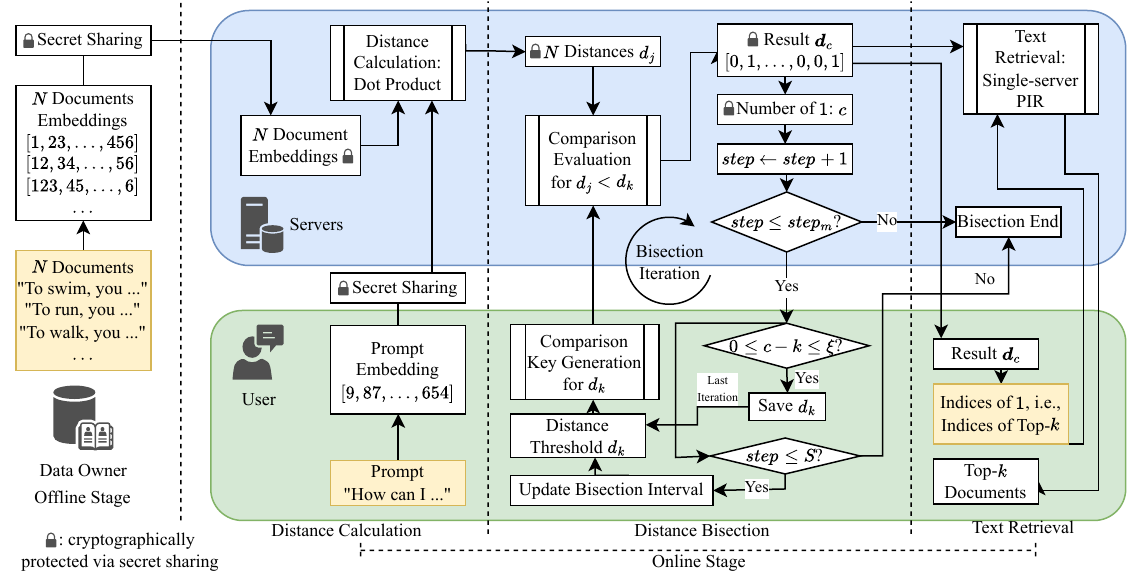}
    \caption{
        The workflow of {\myProj}. During the offline stage, the data owner sets up the secret-shared database. During the distance calculation, the servers compute the secret-shared distances between each document and the user's prompt. During the distance bisection, the user determines a distance threshold $d_k$ for the top-$k$ documents. The bisection iteration ends when $d_k$ is found, or the number of iterations exceeds an upper bound. During the text retrieval, the user retrieves textual documents using the indices of the top-$k$ documents.
    }
    \label{fig:workflow}
    \end{center}
\end{figure*}

\subsection{Distance Calculation}
\label{sec:design_distance}

The servers compute the cosine distance $[d_j]$ between the prompt embedding and each document embedding.
The document embedding must be $\ell_2$-normalized by the data owner.
We assume the $\ell_2$-norms of the database $l = \|\vx_j\|$ and the prompt $l_u = \|\vp\|$ are public because they do not leak any information about the database or the prompt.
Note that the dot product $d_j = \vp \cdot \vx_j$ differs from the cosine distance only by the factor $l l_u$.
The dot product outputs the same ranking as the cosine distance because the factor is constant.
Therefore, we use the dot product as $d_j$.
We use \yrcite{prod} triples for each dimension to compute $[d_j] = [\vp \cdot \vx_j]$ from $[\vp]$ and $[\vx_j]$ as \cref{alg:dot_prod}.
This algorithm has $O(N)$ communication costs, but the amount is $4N$ integers and small, e.g., 16MB for $N = 2^{20}$ and 64-bit integers.
\looseness=-1

Current embedding models output floating-point numbers that are $\ell_2$-normalized to $1$.
To represent the embedding elements in $\sF_p$, the data owner scales them by a factor $2^f$ and truncates the fractional parts \cite{scale}, resulting in integers with precision $f$.
While the truncation causes the embeddings' norms to deviate slightly from $2^f$, this error is negligible, as we show in \cref{sec:exp_acc}.
This scale requires that the cardinality of $\sF_p$ must be greater than $2^{f + 1}$.
Note that the product of two scaled integers has $2f$ precision.
Therefore, we are required to truncate every product to avoid causing an overflow for the field $\sF_p$ \cite{prag}.
That is, for integers $a, b$, we use the integer part of $ab / 2^f$ as the product.
\looseness=-1

% Fix cref err
% \cref{def:dcf}
Definition~\ref{def:dcf}
requires the input $a \in \sG_{in}$ to be an unsigned integer because the internal comparison is performed bit by bit, starting from the most significant bit \cite{dcf}.
% use Boyle et~al.'s idea \yrcite{dcf_cmp}
We extend \cref{alg:cmp} to signed integers by giving $\sF_p$ a new bit representation.
The elements $-\lfloor p/2 \rfloor, \cdots, \lfloor p/2 \rfloor$ are represented as the bits of the unsigned integers $0, \cdots, p - 1$.
This bit representation transformation does not change any algorithm step.

\subsection{Distance Bisection}
\label{sec:design_bisection}

Given all $[d_j]$, we aim to obtain their top-$k$ values.
The naive method is to let the servers return all $d_j$ to the user, who then sorts them.
However, though no embeddings are given to the user, letting the user know all $N$ distances leaks information about the database and can lead to data recovery \cite{leak_distance,leak_distance2,leak_distance1}.
Existing systems avoid this by sorting $[d_j]$ on the servers, which is inefficient in security protocols.

Unlike existing systems, {\myProj} runs an interactive bisection method between the servers and the user to determine a distance threshold $d_k$.
The user starts with the interval $[-l l_u, l l_u]$ and $d_k = 0$.
During each bisection iteration, the servers compare each distance $[d_j]$ with $[d_k]$ and return to the user the count of distances $c$ that are less than $d_k$.
To compute $[c]$, the servers compare $[d_k]$ with each $[d_j]$ using \cref{alg:cmp} with the interval $[d_k, p)$.
That is, the user runs the key generation method $Cmp.Gen([d_k, p))$, distributing $k_i^l, k_i^r, [w]_i$ to the server $i$, respectively.
The servers run the evaluation method $Cmp.Eval([d_j])$ for each $[d_j]$, resulting in $[0]/[1]$.
All $[0]/[1]$ outputs of all $[d_j]$ form a vector $[\vd_c]$.
The servers sum their elements to $[c]$ and return $c$ to the user.
To avoid carrying, the cardinality of $\sF_p$ must be greater than $N$.
The user compares $c$ with $k$ to update the interval and $d_k$.
When either the user or the servers stop the iteration, as defined below, the servers run one more iteration with a particular $d_k$ saved by the user and return $\vd_c$ instead of $c$ to the user.
\looseness=-1

The servers must limit the number of iterations.
At each iteration, the user receives the number $c$ as leakage.
This small leakage is aggregated from all documents and does not identify any particular document.
However, this leakage can accumulate during users' queries.
The servers limit the number of iterations with an upper bound $step_m$ to limit the total amount of leakage.

The user must run a fixed number of iterations before stopping.
Otherwise, the server can infer information about the user's prompt from the stopping time.
For example, if the number of iterations is large, the user's prompt embedding is more likely to lie in a region with denser document embeddings.
Using a fixed number of iterations makes the user's behavior independent of the prompt.
This number can still depend on $N$ and $k$ because $N$ and $k$ are public.
We set the number of iterations to $S = \lceil \log_2 (N / k) \rceil$ and assume $S \le step_m$.
These $S$ iterations may not always be sufficient to obtain the exact top-$k$ results with bisection.
However, 1) the experiments in \cref{app:iter_num} show that this value is sufficient for all datasets, and 2) the user can continue the bisection in the next query because, under our privacy goal, no server learns information about the prompt during the current query.
\looseness=-1

The user can accept an intermediate iteration result rather than the final result with $c = k$ to speed up the query because: 1) in the bisection method, as the step size becomes very small in later iterations, further steps produce little change.
Therefore, an intermediate iteration may still output a result close to the true solution.
Moreover, 2) recent studies show that within a particular yet large context threshold, such as 64K tokens, retrieving slightly more results in RAG generally outperforms RAG with a smaller $k$ for modern long-context models \cite{lc_good,lc_good1,lc_good2}.
We model this intermediate acceptance as a slack value $\xi$ chosen by the user.
The user can accept the iteration result that has $k \le c \le k + \xi$.
As a result, for $k' = k + \xi$, we change $S$ to $\lceil \log_2 (N / k') \rceil$.
We will keep using this $S$ definition below.
\looseness=-1

After termination, the servers must limit the number of $[1]$ in the final $[\vd_c]$, i.e., limit the maximum $[c]$, because $[\vd_c]$ controls which indices can be retrieved by the user during the text retrieval.
A user can use $d_k = ll_u + 1$ to get a $[\vd_c]$ whose elements are all $1$, resulting in accessing any document.
The servers perform this check by publicly setting up an upper bound $c_m$ and comparing the final $[c]$ with $c_m$ using \cref{alg:cmp}.
Unlike the comparison during iterations, the key generation method of this comparison is executed offline by the trusted dealer.

A malicious user can deviate from the protocol to manipulate $[\vd_c]$.
When the servers return $c$, this manipulation does not give the user advantages over the leakage of $c$.
When the servers return $\vd_c$, the servers must also verify that all elements of $[\vd_c]$ are $[0]/[1]$.
To achieve that, the servers compare each element with $2$ using \cref{alg:cmp}.
The trusted dealer now executes the key generation method for this comparison offline.
\looseness=-1

\subsection{Text Retrieval}
\label{sec:design_text}

Given $\vd_c$, the user retrieves the indices of all its $1$ elements to get documents using Private Information Retrieval (PIR).
Because the database is stored with secret sharing, the user uses single-server PIR \cite{single_pir2,single_pir1} to retrieve the documents from both servers.
Following prior work \cite{kann,prag}, we assume that an existing single-server PIR protocol can be used and do not include it in comparisons with prior work.
We benchmark its performance additionally in \cref{app:text_retrieval}, taking SimplePIR \cite{single_pir} as an instance.
Another requirement is that the servers must verify that the user only retrieves the indices that match $[\vd_c]$.
We observe that in single-server PIR, e.g., SimplePIR \cite{single_pir}, there are operations to multiply the database by the query, where each document text is converted to a number.
We multiply this result by $[\vd_c]$ to enforce the requirement.
\looseness=-1

\section{Security Analysis}
\label{sec:analysis_sec}

In this section, we prove {\myProj}'s security attributes.
Our security analysis focuses on the security goals described in \cref{sec:problem_goal}.
We use the simulation-based method \cite{uc} to prove that {\myProj} achieves the privacy goal.
We quantify the information leakage of the database for the bounded database leakage goal.
Additional security analysis for multiple queries and attacks enabled by the database leakage is available in \cref{app:more_sec}.

\begin{claim}[{\myProj}'s Privacy]
    For a Probabilistic Polynomial Time (PPT) adversary $\mathcal{A}$ corrupting the server $i*$ for $i* \in \{0, 1\}$ but following the protocol, \cref{alg:full_online} guarantees that $\mathcal{A}$ learns no information about the prompt and the database.
    \looseness=-1
\end{claim}

\begin{proof}
    The view of $\mathcal{A}$ includes:
    1) the prompt embedding share $[\vp]_{i*}$;
    2)$p_n - r^{(a)}_n, x_{j,n} - r^{(b)}_n$ for $n \in [0, m)$ during the distance calculation, where $r^{(a)}_n, r^{(b)}_n$ are random for $n \in [0, m)$;
    3) At most $min\{ step_m, \lceil \log_2 (N/(k + \xi)) \rceil \}$ keys $cmpkey^{\ge d_k}$ for different $d_k$, which are output by $Cmp.Gen$ and sent from the user;
    4) $d_j + r^{(d_j)}, c + r^{(c)}, d_{c,j} + r^{(d_{c,j})}$ during $Cmp.Eval$, where $r^{(d_j)}, r^{(c)}, r^{(d_{c,j})}$ are random.

    We construct a PPT simulator $Sim$ that generates a view computationally indistinguishable from what $\mathcal{A}$ learns by attacking the protocol, proving {\myProj}'s privacy goal.
    For 1), 2), and 4), $Sim$ outputs random values, which are indistinguishable because these values are either secret-shared or masked by random values.
    For 3), as each $cmpkey$ has $cmpkey = (k_{i^*}^l, k_{i^*}^r, [w]_{i^*})$, $Sim$ outputs a random value for $[w]_{i^*}$, which is indistinguishable because it is secret-shared.
    $Sim$ uses $Sim_{dcf}$'s output for $k_{i^*}^l, k_{i^*}^r$, which are indistinguishable due to DCF's privacy as described in \cref{sec:preliminary}.
    \looseness=-1
\end{proof}

\textbf{Bounded database leakage}.
We follow \yrcite{leakage} and \yrcite{kann}'s perspectives to divide the database leakage to the user into two incomparable and complementary parts:
1) physical leakage, which is the leakage of {\myProj}'s protocol beyond the top-$k$ results;
2) functional leakage, which is the leakage from the top-$k$ results themselves.
In particular, {\myProj}'s functional leakage is at most $k + \xi$ results, which are $O(k(m + \log N))$ bits of information, regardless of {\myProj}'s protocol \cite{kann}.
This leakage is mentioned as the ``baseline'' leakage in \cref{sec:problem_goal}.
{\myProj}'s physical leakage is at most $S = \lceil \log_2 (N / (k + \xi)) \rceil$ values of $c$ given to the user, where $0 \le c \le N$.
To distinguish between $N + 1$ possible outcomes (from $0$ to $N$), the maximum number of information bits contained in $c$ is $\log_2 (N + 1)$ and $O(\log N)$.
That is, {\myProj}'s physical leakage is $O(\log^2 N)$ bits.
\looseness=-1

We identify the $O(\log^2 N)$ leakage as small, proving {\myProj}'s bounded database leakage goal. Our justification is that this leakage is small compared to the ``baseline'' leakage. We use our experiment settings in \cref{sec:exp_perf}, which match our datasets and the RAG applications. We use these settings so that we can do numerical analysis. We have the following analysis: 1) For the complexity, because $O(\log N)$ ($N =$ 1K $\rightarrow$ 1M) grows slower than $O(k)$ ($k =$ 16 $\rightarrow$ 200), the $O(\log^2 N)$ leakage grows slower than the $O(k(m + \log N))$ ``baseline'' leakage. 2) For the amount, the $O(\log^2 N)$ leakage is $S\log_2(N + 1)$ bits for each query. Meanwhile, the ``baseline'' leakage for each query is $32mk'$ bits, i.e., $k'$ embeddings whose elements are 32-bit. For a setting where $k' = 16$, $m = 1024$, $N = 2^{17}$, and $S = 13$, the ``baseline'' leakage is far larger (more than $1000\times$) than the $O(\log^2 N)$ leakage.

\revision{\reviewer{xsNT-W1,astA-Q1} Count leakage can support targeted attacks, but those attacks stay inside the above bound. For membership inference, a user can submit the embedding of a suspected document and a very small threshold. This tests one target and costs one query. For distance inversion, the user can use bisection to locate the distance of a rank threshold to $O(step_m)$ bits of precision. These attacks reveal aggregate count information, not document identities, document contents, or full distance orderings. Prior database-recovery attacks that rely on returned neighbor IDs, access patterns, or sorted distances therefore cannot be applied directly, because {\myProj} reveals only counts during bisection and uses PIR for text retrieval.}

\section{Experiments}
\label{sec:experiments}

In this section, we benchmark {\myProj}'s accuracy and performance with experiments.
We show that:
1) for accuracy, it has no computational errors, and its results are relevant to the prompt;
2) for performance, it supports arbitrary $k$ and outperforms the state-of-the-art system, PRAG \cite{prag}.
We vary $k$ and the document number $N$.
We keep using $k' = k + \xi$ for the number of retrieved documents, allowing {\myProj} to retrieve a few more results than $k$.
The experiment settings are introduced in \cref{app:exp_settings}.
% \looseness=-1

% \subsection{Implementation Complexity}

% {\myProj}'s implementation has two layers. The cryptographic primitive layer includes the DCF scheme and \cref{alg:cmp}. We provide this layer as a library in the code artifact, with the main API exposed through \texttt{include/fss/dcf.h}. The application logic layer includes the interactive bisection protocol in \cref{alg:full_online}. This layer is about 400 lines of C in our artifact. It calls DCF and comparison routines as black boxes, so an implementer needs to write standard bisection, stopping, and verification logic, but does not need to implement new cryptographic primitives.

\subsection{Accuracy}
\label{sec:exp_acc}

{\myProj}'s retrieval is accurate after computing the distances.
It has small numerical errors introduced by converting embedding elements to integers and by the multiplication of the distance calculation.
We use BEIR's \texttt{trec-covid} dataset \cite{dataset}, which has 171332 documents.
We use BEIR's precomputed embeddings from Cohere/beir-embed-english-v3 \cite{dataset_emb}.
We measure the recall between the results from {\myProj}'s integer distances and the original distances, i.e., the dot products of the floating-point embeddings.
{\myProj}'s recall remains $1$ for $k' = 16$--$1024$, showing {\myProj}'s high computation accuracy.
We omit the recall figure because its curve is constant.
Because {\myProj} only considers cosine and dot product distances, we also measure the results' average relevance score to show whether {\myProj} retrieves the relevant documents, resulting in \cref{fig:acc}.
The relevance score is a human-annotated categorical level of relevance.
$0, 1, 2$ correspond to irrelevant, relevant, and highly relevant, respectively.
{\myProj}'s relevance score is larger than $1$ for $k' \le 300$, meaning all of the retrieved documents are generally relevant.

\begin{figure}[t]
    \begin{center}
        \centerline{\includegraphics[width=\columnwidth]{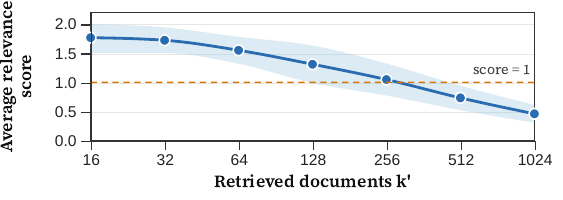}}
        \caption{Relevance score (higher is better) with varying numbers of retrieved documents $k'$}
        \label{fig:acc}
    \end{center}
    \vskip -0.2in
\end{figure}

\subsection{Performance}
\label{sec:exp_perf}

{\myProj}'s performance is independent of the particular document or prompt that is involved in its protocol.
This is a necessary attribute for a security system to defend against the timing attack \cite{timing}.
As a result, we use synthetic datasets for the experiments.
We benchmark {\myProj}'s end-to-end performance.
For \cref{alg:dot_prod} and DCFs used by \cref{alg:cmp}, we use multi-threading and start 96 threads.
Because the baseline PRAG does not support multi-threading, we run a batch of 96 instances and report the amortized time per instance.
\looseness=-1

\textbf{Computation costs}.
\cref{fig:perf} and \cref{tab:perf_exact} show that {\myProj} is 3$\times$ faster than PRAG when $k' = 16$ and changing $N$, and 7--300$\times$ faster when $N = 2^{14}$ and $k' = 32$--$1024$.
Both {\myProj} and PRAG's server time is linear in $N$ and $k'$, but {\myProj} becomes faster with larger $k'$.
We also compare {\myProj} with a non-secure RAG baseline in \cref{app:nonsecure}.
\looseness=-1

\begin{figure}[t]
    \begin{center}
        \centerline{\includegraphics[width=\columnwidth]{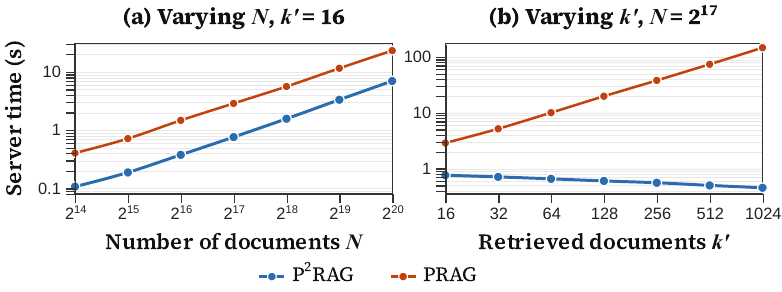}}
        \caption{Server time (lower is better) with varying numbers of documents $N$ and retrieved documents $k'$. $k' = 16$ and $N = 2^{17}$ for the two figures, respectively.}
        \label{fig:perf}
    \end{center}
\end{figure}

\begin{table}[t]
    \caption{Server time behind \cref{fig:perf}.}
    \label{tab:perf_exact}
    \begin{center}
    \begin{small}
    \begin{tabular}{lccc}
        \toprule
        $N, k'$ & {\myProj} & PRAG & Speedup \\
        \midrule
        $2^{14}, 16$ & 0.106s & 0.400s & 3.78$\times$ \\
        $2^{17}, 16$ & 0.751s & 2.84s & 3.78$\times$ \\
        $2^{20}, 16$ & 6.89s & 22.9s & 3.33$\times$ \\
        $2^{17}, 128$ & 0.594s & 19.5s & 32.9$\times$ \\
        $2^{17}, 1024$ & 0.448s & 144s & 322$\times$ \\
        \bottomrule
    \end{tabular}
    \end{small}
    \end{center}
\end{table}

\textbf{Communication costs}.
\cref{tab:comm} shows {\myProj}'s intra-server (IS) and user-server (US) communication volume, together with the number of Round-Trip Time (RTT).
While intra-server communication can reach hundreds of MB, this is well-tolerated via high-bandwidth data center links.
The user-server communication remains manageable, as even the largest one is comparable to only a few standard images.
We do not simulate the network but count the communication volume here because network conditions vary widely in secret-sharing systems and do not affect the system's throughput, as CPUs can run other computational work while waiting for transmission.
As an empirical evaluation for the communication time, we evaluate it in a typical network condition where the servers are in different countries in \cref{app:latency}.

\begin{table}[htb]
    \caption{Communication volume and number of RTTs.}
    \label{tab:comm}
    \begin{center}
    \begin{small}
    \begin{sc}
        \begin{tabular}{lccc}
            \toprule
            Comm. Type & Volume (B) & RTT \\
            \midrule
            User-server & $16384 + 4224 S + 16 N$ & $S + 1$ \\
            Intra-server & $64 N +  16 N S + 32$ & $S + 1$ \\
            \midrule
            $N$, $k'$ & Volume of US, IS & RTT \\
            \midrule
            $2^{17}$, $16$ & $2.168$MB, $35.65$MB & $14$ \\
            $2^{17}$, $128$ & $2.156$MB, $29.36$MB & $11$ \\
            $2^{20}$, $16$ & $16.86$MB, $335.5$MB & $17$ \\
            \bottomrule
        \end{tabular}
    \end{sc}
    \end{small}
    \end{center}
    \vskip -0.2in
\end{table}

\section{Related Work}

\textbf{Privacy-preserving RAG}.
Privacy-preserving RAG has the retrieval and inference phases to be protected.
During the retrieval phase, both the data owner's database and the user's prompt require protection.
{\myProj} belongs to the systems protecting the retrieval phase and protecting both the data owner and the user \cite{prag,he_prag}.
Some other systems make a weaker security assumption, protecting only the user's prompt \cite{remoterag,query_prag,query_prag1} or the data owner's database \cite{data_prag,data_prag1,data_prag2,data_prag3,data_prag4}.
The other systems focus on the inference phase \cite{ppi,ppi1,ppi4,ppi2,ppi3}.
\looseness=-1

\textbf{Cryptographic tools}.
Secure $k$-Approximate Nearest Neighbor \cite{kann0,sanns,he_kann,kann} and secure sorting \cite{secure_sort1,secure_sort2,secure_sort,secure_sort3} provide cryptographic tools for privacy-preserving RAG systems.
{\myProj}'s protocol is adapted from the secure sorting based on Function Secret Sharing (FSS) \cite{fss0,dpf,dcf_cmp,dcf} for its communication and computation efficiency.
\revision{\reviewer{astA-Q3} Classical $k$-th element or quantile selection does not directly solve our setting because plaintext selection methods are data-dependent. In encrypted retrieval, data-dependent branches can leak information unless they are made oblivious, which brings the cost back toward secure sorting or linear top-$k$ maintenance. {\myProj}'s novelty is to move the data-dependent bisection choice to the user while bounding the extra count leakage. This gives a selection-like primitive for encrypted top-$k$ retrieval without sorting all candidates.}
Some other systems use cryptographic or hardware tools such as Differential Privacy (DP) \cite{dp,data_prag1}, Homomorphic Encryption (HE) \cite{he,he_prag}, or Trusted Execution Environment (TEE) \cite{tee,tee_rag}.

\section{Conclusion}

In this paper, we have presented {\myProj}, an efficient privacy-preserving RAG service supporting arbitrary top-$k$ retrieval while protecting both the data owner's database and the user's prompt.
{\myProj} uses secret sharing on two semi-honest non-colluding servers to protect data privacy.
{\myProj} uses an interactive bisection method to select the top-$k$ relevant documents, without sorting candidate documents.
In particular, the bisection method determines a distance threshold by comparing each distance to the threshold using Distributed Comparison Functions (DCFs), while both the threshold and all distances are secret-shared.
This design enables the user to choose an arbitrary $k$, even when $k$ is large.
To defend against malicious users, {\myProj} uses restrictions and verification mechanisms to tightly bound the information leakage of the database.
Both the number of bisection iterations and the result set size are limited.
The experiments show that {\myProj} achieves higher performance with larger $k$ and is 3--300$\times$ faster than the state-of-the-art PRAG for $k = 16$--$1024$.

% Manually achieved
% \newpage

% \section*{Impact Statement}

% This paper focuses on privacy-preserving Retrieval-Augmented Generation (RAG) systems, aiming to enhance their functionality and efficiency while maintaining the same level of data privacy protection.
% The potential broader impact of this work includes advancing the deployment of RAG in privacy-critical domains, such as healthcare and finance, where data privacy is a strict requirement.

\section*{Limitations}

This paper focuses on privacy-preserving Retrieval-Augmented Generation (RAG) systems, aiming to enhance their functionality and efficiency while maintaining the same level of data privacy protection.
When privacy is not critical, standard non-secure RAG is a simpler and more efficient choice because it avoids the cryptographic overhead introduced by {\myProj}.
Therefore, {\myProj} is mainly useful for applications such as healthcare, finance, and legal services, where data privacy is a strict requirement.
\looseness=-1

\bibliography{ref}

\appendix
\crefalias{section}{appendix}

\section{Full Protocol}
\label{app:protocol}

In this section, we summarize \cref{sec:design} as {\myProj}'s full protocol.
{\myProj}'s protocol has the offline and online stages.
The offline stage includes setting up the data owner and servers.
We do not benchmark its performance because it does not affect the online performance of a user's query.
The online stage involves computation and communication between the servers and the user.

During the data owner's offline stage, the data owner has a database filled with textual documents and aims to use {\myProj}.
It pads all documents to the same length and publishes this length, which is typically required by the text retrieval implemented as single-server PIR.
It uses a public embedding model to transform each document to an embedding $\vx$.
The embedding has $m$ dimensions.
Its elements are in a prime field $\sF_p$ as described in \cref{sec:design_distance}.
The document embeddings are $\ell_2$-normalized to $l$.
The data owner publishes $l$.
The data owner shares each document embedding as $[\vx]$ and sends the shares to the servers, respectively.

During the servers' offline stage, they use a trusted dealer to run the offline stage of the dot product method.
They are configured with an iteration step upper bound $step_m$ and an upper bound $c_m$ for $c$.
They use a trusted dealer to run $Cmp.Gen$ of \cref{alg:cmp} with the intervals $[0, 2), [0, c_m + 1)$, resulting in $cmpkey^{<2}/cmpkey^{\le c_m} = (k_i^l, k_i^r, [w]_i)$ for each server $i$.
\looseness=-1

The user has a textual prompt and aims to retrieve top-$k$ relevant documents.
It uses the same public embedding model as the data owner to transform the prompt to an embedding $\vp$.
It chooses a slack value $\xi$ for $k$.
It computes $k' = k + \xi$ and $S = \lceil \log_2 (N / k') \rceil$.

The user shares the embedding as $[\vp]$ and sends $[\vp]$ to the servers, respectively, to start the online stage.
The online stage is shown as \cref{alg:full_online}.
The user receives at least $k$ and at most $k + \xi$ indices from $\vd_c$ as the results.
The user can then perform the text retrieval as described in \cref{sec:design_text} to retrieve the corresponding textual documents from these indices.
\looseness=-1

\begin{algorithm}[htb]
    \caption{Online stage of {\myProj}'s protocol}
    \label{alg:full_online}
    \begin{algorithmic}
        \STATE The servers compute $[d_j] = [\vp \cdot \vx_j]$ for $j \in [0, N)$.
        \STATE The user sets $d_k = 0$ and $[d_l, d_r] = [-ll_u, ll_u]$.
        \FOR{$step = 0$ to $S - 1$}
            \STATE The user runs $Cmp.Gen([d_k, p))$ of \cref{alg:cmp}, resulting in $cmpkey^{\ge d_k} = (k_i^l, k_i^r, [w]_i)$.
            \STATE The user sends $cmpkey^{\ge d_k}$ to the servers, respectively.%
            \STATE The servers run $Cmp.Eval([d_j])$ of \cref{alg:cmp} for $cmpkey^{\ge d_k}$ to get $\vd_c$ for $j \in [0, N)$.
            \STATE $[c] = \sum_j^N [d_{c,j}]$.
            \IF{$step \ge step_m$}
                \STATE \textbf{break}
            \ENDIF
            \STATE The servers return $c$ to the user.
            \IF{$k \le c \le k'$}
                \STATE The user saves $d_k$.
            \ELSIF{$c > k$}
                \STATE $d_l \gets d_k, d_k \gets (d_l + d_r) / 2$.
            \ELSE
                \STATE $d_r \gets d_k, d_k \gets (d_l + d_r) / 2$.
            \ENDIF
        \ENDFOR
        \STATE The user restores the saved $d_k$, runs the above loop one more time, and tells the servers to return $\vd_c$ instead of $c$.
        \STATE The servers run $Cmp.Eval([d_{c,j}])$ with $cmpkey^{<2}$ for each element of $\vd_c$.
        \STATE The servers publish $cmpkey^{<2}$'s results.
        \IF{there is any $0$ result}
            \STATE \textbf{goto abort}
        \ENDIF
        \STATE The servers run $Cmp.Eval([c])$ with $cmpkey^{\le c_m}$.
        \STATE The servers publish $cmpkey^{\le c_m}$'s results.
        \IF{there is any $0$ result}
            \STATE \textbf{goto abort}
        \ENDIF
        \STATE The servers return $\vd_c$ to the user. \textbf{return}
        \STATE \textbf{abort}: The servers frame the user as malicious and abort.
    \end{algorithmic}
\end{algorithm}

\section{Performance Analysis}
\label{app:analysis_perf}

In this section, we analyze {\myProj}'s performance theoretically.
Our performance analysis covers {\myProj}'s online stage, which starts from the user's prompt embedding and ends with the indices of the top-$k$ documents.
We quantify the analysis with the computation and communication complexity.

\textbf{Computation costs}.
For the primitives of \cref{sec:preliminary}, both DCF key generation and evaluation have $O(\lambda \log N)$ \cite{dcf}, and so do the two methods of the comparison of \cref{alg:cmp}.
In {\myProj}'s online stage, the dot product method has $O(m)$ and is executed $N$ times by the servers.
The bisection iteration is executed at most $O(\log N)$ times.
$Cmp.Gen$ is executed $S$ times by the user.
$Cmp.Eval$ is executed $SN + N + 1$ times for each server.
Therefore, the online stage has $O(\lambda \log^2 N)$ for the user and $O(mN + \lambda N \log^2 N)$ for the servers.
\looseness=-1

\textbf{Communication costs}.
For the primitives of \cref{sec:preliminary}, each key generated by DCFs has $O(\lambda \log N)$ \cite{dcf}, and so do the keys of the comparison.
The evaluation method of the comparison of \cref{alg:cmp} has $O(1)$.
In {\myProj}'s online stage, the user sends the prompt embedding share and $2S$ comparison keys, and receives $S$ counts and a share of $\vd_c$ from each server.
The dot product method has $O(m)$.
The servers run it $m$ times and $Cmp.Eval$ $N$ times as the communication costs.
Therefore, the online stage has $O(m + N + \lambda \log^2 N)$ between the user and the servers, and $O(mN + N \log N)$ between the servers.
\looseness=-1

\section{Dot Product Protocol}

We introduce the full version of the dot product protocol used in \cref{sec:design_distance} as \cref{alg:dot_prod}.

\begin{algorithm}[htb]
    \caption{Dot Product: $[c] = [\va \cdot \vb]$}
    \label{alg:dot_prod}
    \begin{algorithmic}
        \STATE \textbf{Offline Stage}:
        \STATE Sample $m$ random $r^{(a)}_n, r^{(b)}_n, r^{(ab)}_n$ for $n \in [0, m)$ s.t. $r^{(a)}_n \cdot r^{(b)}_n = r^{(ab)}_n$. Distribute $[r^{(a)}_n], [r^{(b)}_n], [r^{(ab)}_n]$.
        \STATE
        \STATE \textbf{Online Stage}:
        \FOR{$k = 0$ to $m$}
            \STATE $[d_k] = [a_k] - r^{(a)}_n, [e_k] = [b_k] - [r^{(b)}_n]$.
            \STATE Publish $d_k, e_k$.
            \STATE $[c_k]_i = [r^{(ab)}_n]_i + e_k \cdot [r^{(a)}_n] + d_k \cdot [r^{(b)}_n] + i \cdot d_k \cdot e_k$.
        \ENDFOR
        \STATE $[c] = \sum_{k = 0}^m [c_k]$. Output $[c]$.
    \end{algorithmic}
\end{algorithm}

\section{Additional Security Analysis}
\label{app:more_sec}

In this section, we analyze two additional security aspects: leakage across multiple queries and attacks enabled by the ``baseline'' leakage.
In both cases, the behavior remains within {\myProj}'s stated security goals.

When a user sends multiple queries to {\myProj}, {\myProj}'s privacy ensures that the servers learn no information about the queries themselves during the protocol.
However, the relationships among the queries, e.g., timing patterns, may reveal information about the queries outside the protocol.
Protecting this leakage is outside the scope of this work.
Existing metadata-private messaging systems \cite{express,vuvuzela,chen,nemesis} can be deployed in front of {\myProj} to hide the patterns.
Meanwhile, the database leakage to the user can accumulate additively across multiple queries.
This does not change the per-query bound: the total leakage grows additively with the number of queries and remains measurable.
This attribute is useful in deployment, e.g., the servers can maintain a leakage budget for each user and take action, such as refusing further queries, when the estimated leakage approaches a predefined threshold.
Therefore, multiple queries do not downgrade {\myProj}'s security guarantees.

{\myProj}'s ``baseline'' leakage can support the following two attacks.
A membership inference attack determines whether a specific document exists in the database, and a database distance inversion attack recovers distances between the query and the document results.
The attacks are possible because of the ``baseline'' leakage: results are returned to the user regardless of the protocol.
However, the leakage used by the attacks is still included in {\myProj}'s bounded database leakage analysis.
Therefore, the attacks do not violate our stated security goals.

\section{Experiment Settings}
\label{app:exp_settings}

We run {\myProj}'s experiments on a physical on-premises server.
The server has two AMD EPYC 7352 CPUs.
Each CPU has 24 cores with hyper-threading, for a total of 96 logical cores.
The server's memory is 188 GiB, which is sufficient for the experiments.
Its operating system is Ubuntu 24.04.
We set the embedding dimension $m = 1024$, matching the following used datasets.
We save embedding elements as 64-bit integers.
We set $p = 2^{64} - 59$, which is the largest prime number that fits within 64 bits.
We set $f = 32$, matching the precision of 32-bit floating-point numbers.
We set $\lambda = 128$ for DCFs.
We use Matyas-Meyer-Oseas one-way compression functions with AES-128 as cryptographically secure PRGs of DCFs for speed \cite{owcf}.
We use OpenSSL 3 for AES-128.
We use myl7/fss as the library for DCFs.
All data points are the mean of 7 independent runs, unless explicitly reported as the minimum or maximum.

\section{Iteration Number Experiment}
\label{app:iter_num}

\cref{tab:num_of_iterations} shows that $S$, which is the number of bisection iterations and is designed in \cref{sec:design_bisection}, is enough for the datasets.
We use 5 BEIR datasets to evaluate.
The maximum empirical number of iterations is still less than or equal to $S$, and the average number is much less than $S$.

\begin{table}[htb]
    \caption{Empirical number of iterations with $k' = 16$.}
    \label{tab:num_of_iterations}
    \begin{center}
    \begin{small}
    \begin{tabular}{lcccc}
        \toprule
        Dataset & $N$ & Avg. & Max & $S$ \\
        \midrule
        \texttt{nfcorpus} & $3{,}633$ & 6.2 & 8 & 8 \\
        \texttt{fiqa} & $57{,}638$ & 6.2 & 12 & 12 \\
        \texttt{trec-covid} & $171{,}332$ & 6.8 & 14 & 14 \\
        \texttt{quora} & $522{,}931$ & 6.2 & 15 & 15 \\
        \texttt{trec-news} & $594{,}977$ & 6.1 & 16 & 16 \\
        \bottomrule
    \end{tabular}
    \end{small}
    \end{center}
\end{table}

\section{Communication Time Experiment}
\label{app:latency}

We evaluate the communication time on a real deployment.
The servers are in GCP \texttt{asia-southeast1-c} (Singapore) and \texttt{asia-east2-c} (Hong Kong). The user is in Hong Kong.
\cref{tab:net_setup} shows the latency and bandwidth among them.
According to \cref{tab:comm}, the total communication time is 7.17s in the slowest deployment, where $(N, k')=(2^{20}, 16)$, and 1.45s in the fastest deployment when excluding $k' = 1024$, where $(N, k')=(2^{17}, 128)$.
Even after including this communication time, {\myProj} still outperforms the PRAG baseline.

\begin{table}[htb]
    \caption{Network setup: latency and bandwidth.}
    \label{tab:net_setup}
    \begin{center}
    \begin{small}
        \begin{tabular}{lcc}
            \toprule
            Link & RTT & Bandwidth \\
            \midrule
            Intra-server (SG $\leftrightarrow$ HK) & 33.3ms & 640Mbps \\
            User $\leftrightarrow$ Server (SG) & 40.2ms & 40.7Mbps \\
            User $\leftrightarrow$ Server (HK) & 8.14ms & 58.0Mbps \\
            \bottomrule
        \end{tabular}
    \end{small}
    \end{center}
\end{table}

\section{Text Retrieval Experiment}
\label{app:text_retrieval}

We instantiate SimplePIR \cite{single_pir} to evaluate the server time for retrieving textual documents in the text retrieval. We compare it with the other parts of {\myProj} and PRAG. The results in \cref{tab:text_retrieval} show that text retrieval is not the dominant part of the total server time. Even in the worst case, $(2^{17}, 128)$, its time is close to the other parts of {\myProj} and much lower than those of PRAG.

\begin{table}[htb]
    \caption{Server time (lower is better) with varying numbers of documents $N$ and retrieved documents $k'$.}
    \label{tab:text_retrieval}
    \begin{center}
    \begin{small}
    \begin{tabular}{lcccccc}
        \toprule
        $N, k'$ & {\myProj} & PRAG & Text. per doc. & Text. \\
        \midrule
        $2^{17}, 16$ & 0.751s & 2.84s & 3.60ms & 57.6ms \\
        $2^{17}, 128$ & 0.594s & 19.5s & 3.60ms & 0.461s \\
        $2^{20}, 16$ & 6.88s & 22.9s & 10.4ms & 0.166s \\
        \bottomrule
    \end{tabular}
    \end{small}
    \end{center}
\end{table}

\section{Non-Secure RAG Baseline Experiment}
\label{app:nonsecure}

We compare {\myProj} with a non-secure RAG baseline to show {\myProj}'s privacy overhead.
The non-secure baseline uses parallel merge sort with the same parallel strategy as PRAG.
It sorts all distances to do top-$k$ retrieval, so its server time does not change with $k'$.
We compare it with {\myProj}'s $k' = 16$ deployment, which is the slowest one of \cref{tab:perf_exact}.
\cref{tab:nonsecure} shows the results.

\begin{table}[htb]
    \caption{Overhead over non-secure baseline.}
    \label{tab:nonsecure}
    \begin{center}
    \begin{small}
    \begin{tabular}{lccc}
        \toprule
        $N$ & Non-secure & {\myProj} ($k'=16$) & Overhead \\
        \midrule
        $2^{14}$ & 12.2ms & 0.106s & 8.69$\times$ \\
        $2^{17}$ & 76.3ms & 0.751s & 9.84$\times$ \\
        $2^{20}$ & 638ms & 6.89s & 10.8$\times$ \\
        \bottomrule
    \end{tabular}
    \end{small}
    \end{center}
\end{table}

\end{document}